\documentclass[11pt]{article}
\usepackage[margin=1in]{geometry}
\usepackage[T1]{fontenc}
\usepackage{lmodern}
\usepackage{microtype}
\usepackage{amsmath,amssymb,amsthm,mathtools}
\usepackage{tikz}
\usepackage{subcaption}
\usepackage{enumitem}
\usepackage{authblk}
\usepackage{comment}
\usepackage[textsize=tiny]{todonotes}
\usepackage{xcolor}
\definecolor{linkblue}{RGB}{0,60,110}
\usepackage[colorlinks=true,allcolors=linkblue,hypertexnames=false]{hyperref}
\newtheorem{theorem}{Theorem}[section]
\newtheorem{lemma}[theorem]{Lemma}

\newtheorem{claim}[theorem]{Claim}
\theoremstyle{plain}
\newtheorem*{definition}{Definition}
\usepackage[nameinlink,noabbrev,capitalize]{cleveref}
\crefname{lemma}{lemma}{lemmas}
\Crefname{lemma}{Lemma}{Lemmas}
\crefname{proposition}{proposition}{propositions}
\Crefname{proposition}{Proposition}{Propositions}
\crefname{claim}{claim}{claims}
\Crefname{claim}{Claim}{Claims}
\crefname{appendix}{appendix}{appendices}
\Crefname{appendix}{Appendix}{Appendices}
\usepackage{thmtools,thm-restate}
\newif\ifrestatinglemma
\DeclareRobustCommand{\restatementprefix}{\ifrestatinglemma restated; \fi}
\newcommand{\poly}{\operatorname{poly}}
\newcommand{\dist}{\operatorname{dist}}
\DeclareMathOperator{\ret}{ret}
\DeclareMathOperator{\dom}{dom}
\newcommand{\E}{\mathbb E}
\newcommand{\Prb}{\mathbb P}
\newcommand{\B}{\mathsf B}

\DeclarePairedDelimiter{\card}{\vert}{\vert}

\DeclarePairedDelimiter{\abs}{|}{|}

\DeclarePairedDelimiter{\floor}{\lfloor}{\rfloor}

\DeclarePairedDelimiter{\set}{\{}{\}}

\newcommand{\eps}{\varepsilon}

\usepackage{xspace}
\newcommand{\LLL}{Lovász Local Lemma\xspace}

\title{LOCAL LLL in $\poly(\log\log n)$ Rounds}
\newcommand{\authorentry}[2]{%
  \begin{minipage}[t]{0.32\textwidth}
    \centering
    {\large #1\par}
    \vspace{2pt}
    {\small #2\par}
  \end{minipage}%
}

\author{%
  \begin{minipage}{0.86\textwidth}
    \centering
    \authorentry{Maxime Flin}{Aalto University}\hfill
    \authorentry{Mohsen Ghaffari}{MIT}\hfill
    \authorentry{Magn\'us M.~Halld\'orsson}{Reykjavik University}%
    \\
    \authorentry{Yannic Maus}{TU Graz}\hspace{0.02\textwidth}%
    \authorentry{Alexandre Nolin}{T\'el\'ecom SudParis}
  \end{minipage}%
}
\date{}

\begin{document}
\pagenumbering{gobble}
\maketitle

\vfill

\begin{abstract}
We provide a $\poly(\log\log n)$-round randomized LOCAL algorithm for the
constructive Lovász Local Lemma with polynomial criterion $pd^{c} < 1$,
for a sufficiently large constant $c$. The best previous general round complexity was
$\widetilde O(\log n)$. The key ingredient is a local
reduction that amplifies the LLL probability exponent quadratically, while
increasing the dependency radius only by a constant factor.

\par\vspace{\baselineskip}
\noindent\textbf{AI Statement:} The result was proved by a combination of GPT-6 Astra, GPT-6.1 Sol and Opus 5.5 with essentially no human advice. We have rewritten the paper nearly completely for human-centered exposition. We used additional AIs for review and copyediting. The authors take full responsibility for the manuscript's content and the correctness of results.
\end{abstract}

\thispagestyle{empty}

\vfill
\tableofcontents
\newpage
\pagenumbering{arabic}

\section{Introduction}\label{sec:introduction}
One of the central open problems in distributed graph algorithms is the
complexity of solving the Lov\'asz Local Lemma (LLL) problem.
The significance of this problem is in part due to its remarkably broad role
in the LOCAL model of distributed computing (besides its diverse range of
applications in combinatorics and sequential algorithms): as highlighted
by Chang and Pettie~\cite{CP19}, LLL provides an automatic
mechanism for reducing the failure probability of \emph{any} randomized
LOCAL algorithm for \emph{any} locally checkable problem, provided its
local failures satisfy a local union bound with sufficient slack.
The idea is to treat the algorithm's random bits as variables and local
violations as bad events. An LLL algorithm chooses these bits to avoid
every violation, yielding global success with high probability.

The state of the art on LLL's LOCAL round complexity has an exponential
gap between upper and lower bounds: under a polynomial criterion,
general algorithms take $\widetilde O(\log n)$ rounds~\cite{CPS17,Davies_soda23},
while the lower bound is $\Omega(\log\log n)$~\cite{BFHKLRSU_stoc16}.
For some special cases, faster
$\poly(\log\log n)$-round algorithms have been known: constant dependency
degree~\cite{FG17,RG20}, a stronger exponential
criterion~\cite{GHK18,RG20,Davies_soda23}, and applications such as defective
vertex coloring~\cite{FG17,RG20,HMN_disc22,Davies_soda23}, near-optimal edge coloring~\cite{HMN_disc22,Davies_soda23} and optimal vertex coloring \cite{FHJM26}.

In this paper, we solve the general variable LLL under a polynomial criterion
in $\poly(\log\log n)$ rounds. Before stating the formal theorem, let us
recall the model and the problem.

\paragraph{The LOCAL model.}
In the LOCAL model~\cite{linial1992locality}, the network is abstracted as a graph with at most $n$ nodes, with one processor at each node, having its own unique identifier
and independent random bits. Nodes know the common size upper bound $n$,
which is polynomial in the actual network size. Communication proceeds in synchronous rounds:
in each round, each processor can send an arbitrarily large message to each
neighbor and perform unrestricted local computation. The primary measure
of interest is round complexity, the number of rounds until all nodes have
terminated.

\paragraph{The LLL problem.}
Let $\mathcal X$ be a collection of independent random variables and
$\{A_v\}_{v\in V(G)}$ a family of at most $n$ bad events, with $A_v$ determined by
variables $V_v$. A dependency graph $G$ contains an edge between every pair of distinct
events sharing a variable.
Let $d\ge2$ bound its maximum degree and $\Prb[A_v]\le p$ for every $v$.
Nodes also know these upper bounds $p$ and $d$.
The LLL~\cite{EL74} guarantees an assignment avoiding all events when
$ep(d+1)\le1$. In the LOCAL formulation, we communicate on $G$:
vertex $v$ knows $A_v$, its variables and their distributions, and must
output values for $V_v$ that agree with other vertices' outputs and avoid
$A_v$. For distributed algorithms, we focus on a \emph{polynomial criterion},
$pd^c<1$, for a sufficiently large absolute constant $c$.
Such slack suffices for most standard distributed applications, sometimes
through a reformulation, and the precise constant $c$ is usually
secondary~\cite{CPS17,CP19,FG17,GHK18,RG20,Davies_soda23}.
A stronger \emph{exponential criterion}, $p\le2^{-Cd}$, has previously
enabled faster algorithms.

We now state our main result.

\begin{theorem}\label{thm:main}
An $n$-event variable LLL with dependency degree at most $d\ge2$ and
event probabilities at most $p$ satisfying $pd^{1000}<1$ can be solved
in $\poly(\log\log n)$ randomized LOCAL rounds with success probability
at least $1-n^{-\gamma}$, for every constant $\gamma>0$.
\end{theorem}

This write-up prioritizes simplicity and exposition. Therefore, we do not
optimize the exponents in the round complexity and the LLL criterion.
We are aware of substantial improvements to both constants, but defer
them to later versions.

The proof is surprisingly short and relatively elementary. Therefore, we
present it immediately in \Cref{sec:proof}, after a short recap of
preliminaries in \Cref{sec:preliminaries}, and defer further discussion of the context,
significance, and applications to \Cref{sec:context}.

\section{Preliminaries}\label{sec:preliminaries}

\paragraph{Graph notations.}
For a graph $G=(V,E)$, we write $\B_G(v,r)=\{u\in V\mid\dist_G(v,u)\le r\}$ and, for a vertex set
$U$, let $\B_G(U,r)=\bigcup_{u\in U}\B_G(u,r)$.
The graph $G^r$ joins distinct vertices at $G$-distance at most $r$.
For an integer $t\ge1$, we say that a set $S$ is \emph{$t$-independent in $H$} if
it is independent in $H^t$, equivalently, every pair of distinct vertices
in $S$ has $H$-distance at least $t+1$.
We call $S$ \emph{$t'$-ruling for a set $R$ in $H$} if every vertex of $R$ is
within $H$-distance $t'$ of some vertex in $S$.
Throughout, $n$ may be a global size upper bound, by padding with isolated
vertices for the analysis.
With high probability (w.h.p.) means failure probability at most $n^{-a}$
for any desired constant $a>0$.

\paragraph{LLL notations.}
Let $\mathcal X$ be the collection of independent variables and
$V_v\subseteq\mathcal X$ the variable set of event $A_v$.
A \emph{full assignment} assigns a value to every variable in $\mathcal X$.
A \emph{partial assignment} assigns values to a subset
$\mathcal X'\subseteq\mathcal X$ and leaves the remaining variables unassigned.
An event is \emph{violated} under a full assignment if it occurs, and
\emph{unviolated} otherwise.
A \emph{solution} of an LLL problem $\Pi$ is a full assignment under which
none of its events $A_v$ is violated.
We write $X$ for a random full assignment sampled from the product
distribution of $\mathcal X$, and $X_U$ for its restriction to a variable set
$U\subseteq\mathcal X$. Lowercase $x$ denotes a fixed assignment.
A dependency graph $G$ contains an edge between every pair of distinct
events $A_v$ and $A_u$ whose
variable sets $V_v$ and $V_u$ intersect.

\begin{restatable}[\restatementprefix Ruling sets under bounded local independence; a corollary of~\cite{GV_RulingSetBoundedGrowth_podc07,SEW13,Ghaffari_soda16,RG20}]{lemma}{LemmaBoundedIndependenceRulingSets}\label[lemma]{lem:ruling}
Consider an $n$-node graph $H$ where each ball $\B_H(u,4)$ contains at most $k$ independent vertices. An independent $2$-ruling set can
be found in $\poly(\log k, \log\log n)$ rounds w.h.p.
\end{restatable}
We start with an independent $O(\log\log n)$-ruling set and refine it
using MIS computations on graphs of degree $O(k^2)$. The proof is
deferred to \Cref{sec:proof-ruling}.

\begin{restatable}[\restatementprefix Exponential-criterion LLL; a corollary of~\cite{CPS17,GHK18,RG20,Davies_soda23}]{lemma}{LemmaExponentialLLL}
\label[lemma]{lem:exponential}
There is a constant $C$ such that any $n$-event variable LLL with dependency
degree at most $d\ge2$ and event probabilities at most $p\le2^{-Cd}$
can be solved in $\poly(\log\log n)$ rounds w.h.p.
\end{restatable}
This follows from the retraction-based LLL solver of
Ghaffari--Harris--Kuhn~\cite[Section 6]{GHK18}, borrowing the idea from
Chung--Pettie--Su~\cite[Section 3.3]{CPS17}, combined with the network
decomposition of Rozho\v{n}--Ghaffari~\cite{RG20}. The bound was later
made explicit by Davies~\cite[Theorem 19, $r=1$]{Davies_soda23}.
To keep the LLL aspects self-contained, we include a proof in
\Cref{sec:proof-exponential}.

\section{Proof}\label{sec:proof}

\subsection{Proof outline}\label{sec:outline}

At a high level, we combine an LLL exponent amplification
in the spirit of Fischer--Ghaffari~\cite{FG17}, with control of conditional
probabilities under reassignments, in the style of
Chung--Pettie--Su~\cite{CPS17},
Ghaffari--Harris--Kuhn~\cite{GHK18}, and
Davies~\cite{Davies_soda23}. However, the specific constructions and arguments here are quite different.

The core ingredient is the following LLL exponent amplification,
formalized in \Cref{lem:amplification}: given an LLL
problem $\Pi$ on $G$ with event probabilities at most $p=d^{-\lambda}$,
we construct a helper LLL problem $\Pi'$ on the same underlying variables,
with dependency graph $G'=G^{O(1)}$, degree bound $d'=d^{O(1)}$,
and event probabilities at most $p'\le d^{-\Omega(\lambda^2)}$.
Hence, we can define a new parameter $\lambda'$ by
$p'=(d')^{-\lambda'}$, where $\lambda'=\Omega(\lambda^2)$.
The LLL problem $\Pi'$ is defined such that \emph{any} solution of $\Pi'$
can be converted, with high probability, into a solution of $\Pi$ in
$\poly(\log\log n)$ additional rounds on $G$. We will return
to how this conversion works shortly.

We now explain how to use this exponent amplification recursively.
We solve problem $\Pi'$ recursively by applying the same lemma with its new parameters.
Each recursion level gives $\lambda'=\Omega(\lambda^2)$,
a quadratic improvement in the LLL exponent, with the caveat that
we move from graph $G$ to graph $G'=G^{O(1)}$.
Simulating one round on $G'$ costs $O(1)$ rounds on $G$.
Starting at $\lambda=1000$, within $O(\log\log\log n)$ recursion
levels, we reach event probability at most $n^{-a}$ for any desired
sufficiently large constant $a$. At that point, just sampling the random variables succeeds
with high probability. Every edge of every recursion graph, including
the last, spans distance at most $2^{O(\log\log\log n)}$ in the
original graph $G$, so the
$\poly(\log\log n)$ rounds of nonrecursive work at each recursion level
translate to $2^{O(\log\log\log n)}\cdot\poly(\log\log n)
=\poly(\log\log n)$ rounds in $G$ per recursion level.
Summing over the $O(\log\log\log n)$ levels still gives
$\poly(\log\log n)$ rounds.

We next outline the construction and solution transformation underlying
\Cref{lem:amplification}.
Given a solution of $\Pi'$, we have an assignment $x$ for $\Pi$
in which some events may be violated, but with this strong guarantee
built into the definition of $\Pi'$: every ball of a constant radius
contains only $O(\lambda)$ vertices in any $2$-independent set of
\emph{dangerous events}. (We refer to this as \emph{local bounded independence}.)
Informally, an event is dangerous if erasing nearby
variable assignments can make its marginal probability too large.
The definition implies that every violated event is dangerous.
Then, using \Cref{lem:ruling}, we choose a $2$-independent $O(1)$-ruling set of
dangerous vertices, which we call \emph{centers}, and erase the values
in $x$ of all variables appearing in events within a constant-radius
neighborhood of these centers. Each event now has small probability
when the erased variables are sampled independently, and events with no erased variables are not
violated. We group the affected events by nearby centers and cast the
problem of determining values for the erased variables as a new LLL problem,
which we call the \emph{repair-LLL}. Avoiding the new event for each center
in this repair-LLL ensures that none of the original events in $\Pi$ grouped with this center
occurs. Because of the local bounded independence mentioned above, one can
see that this problem has a dependency degree bound of $O(\lambda)$, while its event probabilities are at most
$d^{-\Omega(\lambda)}$. These parameters
satisfy the stronger exponential criterion. Therefore, we can solve this
problem directly in $\poly(\log\log n)$ rounds by invoking a known
fast LLL algorithm (of Ghaffari--Harris--Kuhn~\cite{GHK18} or
Davies~\cite{Davies_soda23}; \Cref{lem:exponential}).

\subsection{LLL exponent amplification via dangerous events}\label{sec:danger}

\begin{lemma}[LLL exponent amplification]\label[lemma]{lem:amplification}
There is an absolute constant $d_{\min}\ge3$ such that the following
holds for any fixed constant $a>0$.
Let $\Pi$ be an $n$-event variable LLL problem with dependency graph
$G$, dependency degree upper bound $d\ge d_{\min}$, and event
probabilities at most $p=d^{-\lambda}$, where
$1000\le\lambda\le a\log n$.
In $O(1)$ LOCAL rounds on $G$, we can construct a new variable LLL
problem $\Pi'$ on the same underlying independent variables, with
one event per vertex of $G$, dependency graph $G'$, dependency degree
upper bound $d'$, and event probability upper bound $p'$ satisfying
\[
 G'=G^{23},\qquad d'=d^{23},\qquad
 p'=(d')^{-\lambda'}\le d^{-23\lambda^2/700}.
 \quad\text{Thus, we may take }\lambda'=\frac{\lambda^2}{700}.
\]
Any satisfying assignment for $\Pi'$ can be transformed into a
satisfying assignment for $\Pi$, with high probability, in
$\poly(\log\log n)$ randomized LOCAL rounds on $G$.
\end{lemma}

\begin{proof}
Let $C$ be the constant in \Cref{lem:exponential}, and choose
$d_{\min}=\max\{3,2^{3C}\}$.
The proof has three parts: In Part I, we define a basic concept, dangerous
events, and establish its properties. In Part II, we construct the new
LLL $\Pi'$ using this concept. In Part III, we show how to transform any
solution of $\Pi'$ into a solution of $\Pi$.

\subsubsection*{Part I --- Dangerous events and their properties}
We want to identify events whose probability can become too large after
nearby variable values are erased. We use some notation and terminology based on
Ghaffari--Harris--Kuhn~\cite[Definition 6.1]{GHK18}, although our definition
of a dangerous event differs from theirs. For a partial
assignment $y$, let $\Prb_y[A_v]$ denote the \emph{marginal probability}
of $A_v$: assigned variables retain their values in $y$, and unassigned
variables are sampled independently from their original distributions.
For a full assignment $x$ and an event set $U$, the \emph{retraction}
$\ret(x,U)$ erases the values of all variables in $\bigcup_{u\in U}V_u$
and keeps $x$ elsewhere. We call $\Prb_{\ret(x,U)}[A_v]$ the
\emph{marginal under retraction}, or \emph{retracted-marginal} for short.
Thus, for a set $T$ of retraction centers,
$\Prb_{\ret(x,\B_G(T,4))}[A_v]$ is the retracted-marginal of $A_v$ after erasing
variable values in events within distance $4$ of $T$.

\begin{definition}[Dangerous event]
Set $k=\lfloor\lambda/16\rfloor$ and $q=d^{-\lambda/32}$.
Given a full assignment $x$, an event $A_v$ is
\emph{$q$-dangerous with respect to $x$} if
\[
 \Prb_{\ret(x,\B_G(T,4))}[A_v]>q
 \quad\text{for some }T\subseteq\B_G(v,5)\text{ with }|T|\le k.
\]
That is, there is a set $T\subseteq\B_G(v,5)$ of at most $k$
vertices such that, if we erase the assignments of all variables
appearing in events in $\B_G(T,4)$ and sample the erased
variables independently from their original distributions, keeping
all other values fixed as in $x$, the probability of $A_v$ exceeds $q$.
We omit $q$ and $x$ when they are clear from context. Also, for an
event $A_v$ that is dangerous, we also call its vertex $v$ dangerous.
\end{definition}

Let us make some observations about this definition: a nondangerous event has marginal probability at
most $q$ after \emph{every} such retraction, even when the retraction
centers are chosen after seeing $x$. Hence, any event violated by $x$
is dangerous, since the choice $T=\varnothing$ is allowed.
Moreover, for a fixed set $T$, the retraction may erase a nontrivial
subset of the variables of $A_v$, meaning a subset that is neither
empty nor all of $V_v$. If $\dist_G(v,T)\le4$, then all variables in
$V_v$ are erased, and thus $\Prb_{\ret(x,\B_G(T,4))}[A_v]=\Prb[A_v]\le p$.
However, for instance, if $T=\{v'\}$ with $\dist_G(v,v')=5$, then $v\notin\B_G(T,4)$,
but variables shared with events in $\B_G(T,4)$ are erased,
while other variables of $A_v$ may retain their values in $x$.
These retained values can make $\Prb_{\ret(x,\B_G(T,4))}[A_v]$ much larger than $p$.

\begin{claim}\label[claim]{clm:dangerous-probability}
Let $D_v$ denote the event that $A_v$ is dangerous under the random assignment $X$.
For every $v$, $\Prb[D_v]\le\rho:=d^{-19\lambda/32}$.
Moreover, for each fixed $T\subseteq\B_G(v,5)$ with $|T|\le k$, the random variable $\Prb_{\ret(X,\B_G(T,4))}[A_v]$
depends only on $X_{V_v}$. Hence, the event
$\{\Prb_{\ret(X,\B_G(T,4))}[A_v]>q\}$ also depends only on $X_{V_v}$.
Consequently, for any $2$-independent set $S$, the events $D_v$,
$v\in S$, are mutually independent and
$\Prb[\bigcap_{v\in S}D_v]\le\rho^{|S|}$.
\end{claim}

\begin{proof}
For fixed $T$, $\E[\Prb_{\ret(X,\B_G(T,4))}[A_v]]=\Prb[A_v]\leq p=d^{-\lambda}$,
since averaging over the retained values and resampling the erased
variables recovers the original product distribution.
Thus, Markov's inequality gives $\Prb[\Prb_{\ret(X,\B_G(T,4))}[A_v]>q]\le p/q$.
Since $d\ge3$, we have $|\B_G(v,r)|\le d^r$ for $r\ge3$.
In particular, $|\B_G(v,5)|\le d^5$, so the number of sets $T$ is
at most $(|\B_G(v,5)|+1)^k\le d^{6k}$. A union bound over these sets gives
\begin{equation}\label{eq:dangerous-probability}
 \Prb[D_v]\le d^{6k}p/q\le d^{6k}d^{-\lambda}/q\le\rho.
\end{equation}
For each fixed $T$, $\Prb_{\ret(X,\B_G(T,4))}[A_v]$ depends only on $X_{V_v}$:
the erased variable set is deterministic, and $A_v$ uses only $V_v$.
The family of sets $T$ under consideration is also deterministic, so $D_v$
depends only on $X_{V_v}$. For any $2$-independent set $S$, the
variable sets $V_v$, $v\in S$, are pairwise disjoint.
Thus, the events $D_v$, $v\in S$, are mutually independent, and
$\Prb[\bigcap_{v\in S}D_v]=\prod_{v\in S}\Prb[D_v]\le\rho^{|S|}$.
\end{proof}

\subsubsection*{Part II --- Constructing \texorpdfstring{$\Pi'$}{Pi'} and bounding its parameters}
\paragraph{Defining the new LLL \texorpdfstring{$\Pi'$}{Pi'}.}
Define $\Pi'$ to be the problem of finding an assignment $x$ to the
variables of $\Pi$ satisfying
\begin{equation}\tag{$\star$}\label{eq:target}
 \textit{each }\B_G(v,11)\textit{ contains no $2$-independent set of }k+1
 \textit{ dangerous vertices},
\end{equation}
where dangerous is with respect to $x$.
We refer to this property as \emph{local bounded independence} of the dangerous vertices.
To express $\Pi'$ as an LLL, let $F_v$ be the event that
$\B_G(v,11)$ contains a $2$-independent set of $k+1$ dangerous vertices.
Thus, $\Pi'=\{F_v\}$: avoiding all $F_v$ is equivalent to
Condition~\eqref{eq:target}.

We note that Condition~\eqref{eq:target} ensures that the repair-LLL,
which we introduce in Part III and use to obtain a solution of $\Pi$ from
the given solution of $\Pi'$, has dependency degree at most $k$.
This ensures that this LLL satisfies the exponential criterion and can be
solved directly in $\poly(\log\log n)$ rounds, since its event probabilities
are at most $d^{-\Omega(k)}\le2^{-Ck}$, as we show in Part III.\footnote{Let us also comment on the role and balancing of the parameter $k$.
Larger $k$ buys stronger amplification, as it directly enters the failure
exponents for events of $\Pi'$ in \eqref{eq:amplified-probability}.
However, it makes dangerous events more likely in \eqref{eq:dangerous-probability},
where we can only afford $k\le O(\lambda)$,
and makes the repair-LLL's exponential criterion harder to meet, since $k$
bounds its dependency degree while \eqref{eq:completion-probability} bounds
its event probabilities.}

\paragraph{The parameters of the new LLL \texorpdfstring{$\Pi'$}{Pi'}.}
We use $p'$ and $d'$ to denote the event probability
and dependency degree bounds of $\Pi'$, respectively. We next bound these. 

We start with $d'$. For each event $F_v$, its variable set is
$\bigcup_{u\in\B_G(v,11)}V_u$.
There is a dependency edge between $F_v$ and $F_w$ if
$\left(\bigcup_{u\in\B_G(v,11)}V_u\right)
\cap \left(\bigcup_{u\in\B_G(w,11)}V_u\right)\ne\varnothing$,
which implies $\dist_G(v,w)\le11+1+11=23$.
Hence, we may take $G'=G^{23}$ as the dependency graph.
Its dependency degree is at most $|\B_G(v,23)|-1\le d^{23}$,
so we may take $d'=d^{23}$ as its dependency degree upper bound.

We now bound $p'$. Since $|\B_G(v,11)|\le d^{11}$, \Cref{clm:dangerous-probability} and a union bound give, for every $v$,
\begin{equation}\label{eq:amplified-probability}
 \Prb[F_v]
 \leq \sum_{\stackrel{S \in \binom{\B_G(v, 11)}{k+1}}{\text{$S$ is 2-independent}}} \Prb\left[ \bigcap_{u\in S} D_u \right]
 \le d^{11(k+1)}\rho^{k+1}
 =d^{-(19\lambda/32-11)(k+1)}
 \le(d^{23})^{-\lambda^2/700}.
\end{equation}
The last inequality follows, for $\lambda\ge1000$, from
$(19\lambda/32-11)(k+1)>19\lambda^2/512-11\lambda/16
\ge23\lambda^2/700$.
Thus, we may take $p'=(d^{23})^{-\lambda^2/700}$, which means we may
set $\lambda'=\lambda^2/700$ for $d'=d^{23}$.
This proves the claimed amplification from the $(d,p,\lambda)$-LLL
$\Pi$, where $p=d^{-\lambda}$, to the $(d',p',\lambda')$-LLL $\Pi'$, where
$d'=d^{23}$, $\lambda'=\lambda^2/700$, and $p'=(d')^{-\lambda'}$.

The LLL problem $\Pi'$ can be defined in $O(1)$ rounds: radius $20=11+5+4$
suffices to describe all the retracted-marginals defining $F_v$.
The problem $\Pi'$ is defined on $G'=G^{23}$, so one round in $G'$
costs $23$ rounds in $G$.

\subsubsection*{Part III --- Transforming a solution of \texorpdfstring{$\Pi'$}{Pi'} into a solution of \texorpdfstring{$\Pi$}{Pi}}
Consider an arbitrary solution $x$
of $\Pi'$. Thus, $x$ satisfies Condition~\eqref{eq:target}. Notice that the variable assignments in $x$ may violate some events of $\Pi$.
Our remaining task is to compute a solution of $\Pi$, based on $x$. 
For that, the plan outline is to cover dangerous events by centers, erase nearby variable values, and determine new
values for these erased variables by setting up a new repair-LLL based on
these centers so that it satisfies the exponential criterion, and solving
it directly with \Cref{lem:exponential}. This gives assignments for the
erased variables so that no event of $\Pi$ is violated w.h.p.
Next, we make this concrete.

\paragraph{Choosing centers to cover all dangerous events.}
Let $\mathcal{D}=\{v:x\in D_v\}$ be the dangerous vertices of $\Pi$ under $x$.
If $\mathcal{D}=\varnothing$, $x$ solves $\Pi$. This is because, if there were a
violated event $A_v$ under the assignment $x$, then setting $T=\varnothing$
would witness that $A_v$ is dangerous. Otherwise set $H=G^2[\mathcal{D}]$. Its independent sets are $2$-independent in
$G$, and $\B_H(v,4)\subseteq\B_G(v,8)$, so Condition~\eqref{eq:target}
bounds their size in each such ball by $k=\floor{\lambda/16}\le O(\log n)$.
We compute a set $S\subseteq \mathcal{D}$ that is an independent 2-ruling set of $H$ in $\poly(\log\log n)$ rounds w.h.p., by \Cref{lem:ruling}.
Then $S$ is $2$-independent in $G$ and within distance $4$ of every vertex of $\mathcal{D}$.
Thus, $\mathcal{D}\subseteq\B_G(S,4)$; each $H$-round costs two $G$-rounds.
We call the vertices in $S$ \emph{danger centers}.

\paragraph{Retracting near danger centers.}
First, we retract the assignments in $x$ for all variables appearing in
events in $\B_G(S,4)$. That is, we set $y=\ret(x,\B_G(S,4))$.
Thus, all variables in $W=\bigcup_{u\in\B_G(S,4)}V_u$ are erased,
and all other variables retain their values from $x$.
Next, we argue that the retracted-marginal of each event $A_v$ is at most $q$:
if we sample the erased variables in $W$ independently from their original
distributions, each $A_v$ occurs with probability at most $q$, that is,
$\Prb_y[A_v]\le q$.
A dangerous event $v$ has all its variables erased, so
$\Prb_y[A_v]\le d^{-\lambda}\le q$.
For a nondangerous $v$, put
$T_v=S\cap\B_G(v,5)$. Condition~\eqref{eq:target} gives $|T_v|\le k$,
and the retractions $y$ and $\ret(x,\B_G(T_v,4))$ erase exactly the
same variables of $A_v$.
Indeed, any erased variable of $A_v$ appears in some $A_u$ with
$u\in\B_G(s,4)$ for a center $s$. Sharing that variable gives
$\dist_G(u,v)\le1$, so $\dist_G(s,v)\le5$.
Since we have assumed that $v$ is not dangerous,
$\Prb_y[A_v]=\Prb_{\ret(x,\B_G(T_v,4))}[A_v]\le q$.
Finally, note that any event $A_v$ with no erased variables is permanently
unviolated, since $\Prb_y[A_v]\in\{0,1\}$ and $\Prb_y[A_v]\le q<1$.

\paragraph{Completing to a solution of \texorpdfstring{$\Pi$}{Pi}, via an exponential-criterion repair-LLL.}
What remains is to determine the assignments for the variables in
$W=\bigcup_{u\in\B_G(S,4)}V_u$, which are currently unassigned.
We do this by setting up a new LLL, which we call the \emph{repair-LLL}.
Assign every event $A_v$ containing an erased variable to a center
$s\in S\cap\B_G(v,5)$, breaking ties by IDs.
Then, in the repair-LLL, we put one event $E_s$ for each center $s\in S$: it is the union of all
$\Pi$-events assigned to $s$, viewed as an event over the erased variables only.
We also assign each erased variable one owner among the centers whose assigned
events contain it, to ensure consistent sampling and output.

There is a dependency edge between centers $s$ and $s'$ if their events
share an erased variable. This implies
$\dist_G(s,s')\le5+1+5=11$, so we may use $G^{11}[S]$ as the dependency graph.
By the local bounded independence in Condition~\eqref{eq:target} and the $2$-independence of $S$, this graph has degree
at most $k$.
Furthermore, each $E_s$ contains at most $d^5$ constituent events, each
with marginal probability at most $q$, as argued above. Hence,
$\Prb_y[E_s]\le d^5q=d^{5-\lambda/32}\le d^{-\lambda/40}$,
where the last inequality needs only $\lambda\ge960$.
Since $d\ge d_{\min}\ge2^{3C}$ and $k=\lfloor\lambda/16\rfloor$,
we have
\begin{equation}\label{eq:completion-probability}
 \Prb_y[E_s]\le2^{-(\lambda/40)\log_2 d}\le2^{-Ck}.
\end{equation}
Thus, since this repair-LLL has dependency degree bounded by $k$, it satisfies
the exponential criterion, and we can solve it directly by invoking \Cref{lem:exponential}
in $\poly(\log\log n)$ rounds w.h.p.
This algorithm runs on $G^{11}[S]$; each of its rounds can be simulated
in at most $11$ rounds on $G$, so the round complexity remains
$\poly(\log\log n)$.
Finally, avoiding all repair-LLL events $E_s$ solves our original
LLL problem $\Pi$, because every event with an erased variable is avoided,
and all other events remain unviolated.

We also comment that the two invoked lemmas (\Cref{lem:ruling}, used to
choose the danger centers, and \Cref{lem:exponential}, used to solve the
repair-LLL) can each be run
with a sufficiently small inverse-polynomial failure probability, so their
combined failure probability is at most $n^{-a}$ for any desired constant $a>0$.
\end{proof}

\subsection{Proof of \texorpdfstring{\Cref{thm:main}}{Theorem 1.1}}\label{sec:recursion}

\begin{proof}[Proof of \Cref{thm:main}]
First, we set aside the easy constant-degree case. Let $d_{\min}$
be the constant in \Cref{lem:amplification}.
If $d<d_{\min}$, use Fischer--Ghaffari
\cite[Theorem 1.1]{FG17} with the network decomposition of
\cite{RG20}; its criterion holds since $p(ed)^{32}\le e^{32}d^{-968}<1$.
This takes $\poly(\log\log n)$ rounds; see also
\cite[Corollary 3.9]{RG20}.
A flaw in the shattering analysis of \cite{FG17} was rectified in
\cite[Section 5]{GHMN26}.
Henceforth assume $d\ge d_{\min}$.

We solve the LLL problem by iteratively invoking
\Cref{lem:amplification}. Choose a constant $a$ large enough that
the final union bound gives failure probability at most $n^{-\gamma}$.
In iteration $i$, we have an $n$-event problem $\Pi_i$ on $G_i$,
with dependency degree upper bound $d_i$ and event probability
upper bound $p_i=d_i^{-\lambda_i}$. Initially,
$\Pi_0=\Pi$, $G_0=G$, $d_0=d$, $\lambda_0=1000$, and
$p_0=d^{-1000}$, which bounds the original event probabilities
by the theorem's hypothesis.

For each iteration $i$ with $p_i>n^{-a}$, invoke
\Cref{lem:amplification} on $\Pi_i$ and define $\Pi_{i+1}$ to be
the resulting problem $\Pi'$. Its parameters are
\[
 G_{i+1}=G_i^{23},\qquad d_{i+1}=d_i^{23},\qquad
 \lambda_{i+1}=\lambda_i^2/700,\qquad
 p_{i+1}=d_{i+1}^{-\lambda_{i+1}}.
\]
The lemma applies because $d_i\ge d_0\ge d_{\min}$ and
$\lambda_i\ge1000$, while $p_i>n^{-a}$ gives
$\lambda_i\log d_i<a\log n$. Hence, $\lambda_i\le a\log n$.
We recursively solve $\Pi_{i+1}$. By the lemma, any solution of
$\Pi_{i+1}$ can then be transformed into a solution of $\Pi_i$
in $\poly(\log\log n)$ additional rounds on $G_i$, with high
probability. All problems retain the same $n$ vertices and
underlying independent variables.

In the \emph{terminal case}, $p_i\le n^{-a}$, we sample the
variables once from their original distributions. A union bound
over the $n$ events shows that this solves $\Pi_i$ with failure
probability at most $np_i\le n^{1-a}$. We then apply the lemma's
transformations in reverse order to obtain a solution of $\Pi_0$.

To bound the number of iterations, the recurrences give
$G_i=G^{23^i}$, $d_i=d^{23^i}$, and
$\lambda_i=700(10/7)^{2^i}$. Since $\lambda_i\log d_i\ge a\log n$
ensures termination, the first terminal level is $L=O(1+\log\log\log n)$.
Let $T_i$ denote the total time of the recursive algorithm for solving $\Pi_i$,
measured in rounds of $G_i$, and let
$b(n)=\poly(\log\log n)$ bound the construction and transformation
costs in \Cref{lem:amplification}.
Then $T_L=O(1)$ and $T_i\le23T_{i+1}+b(n)$, giving
$T_0=O(23^L b(n))=\poly(\log\log n)$.
Finally, we run each transformation with failure probability at most
$n^{-a}$. Since the lemma can use any solution of $\Pi_{i+1}$ to
solve $\Pi_i$, this bound holds conditional on all preceding steps
succeeding. A union bound
over the $L$ transformations and the terminal sample gives total
failure probability at most $n^{1-a}+Ln^{-a}\le n^{-\gamma}$ for
sufficiently large $a$.
\end{proof}

\section{Further Context, Related Work, and Applications}
\label{sec:context}

\subsection{History of LLL} 

\paragraph{Probabilistic Method.}
The \LLL was introduced by Erd\H{o}s and Lovász \cite{EL74} as a tool to prove the existence of combinatorial objects using the probabilistic method. Their proof was purely existential and applied to abstract events without reference to variables.
Shortly after, Spencer presented a generalization of LLL, credited to Lovász and known as the asymmetric criteria, that allows non-uniform bounds on error probabilities
\cite{spencer1977asymptotic}. Shearer then computed the supremum of the $p$ for which the existential conclusion of LLL remains valid, for a fixed dependency degree $d$ \cite{shearer1985problem}. These results have found many applications in combinatorics; see, e.g., \cite[Chapter 5]{AS08} or \cite[Chapter 4]{MR02}.

\paragraph{Constructive LLL.}
Beck initiated the study of the constructive LLL in which one aims to compute an assignment of an LLL's variables that avoids all bad events \cite{Beck_rsa91a}. In the same year Alon adapted Beck's approach to design efficient deterministic parallel algorithms for LLL type problems, most prominently for hypergraph $2$-coloring~\cite{Alon91}. In two decades various groups of researchers published various algorithms for LLL type problems, e.g.,  Molloy and Reed \cite{MR_stoc98}, Czumaj and Scheideler \cite{CS00}, Srinivasan \cite{Srinivasan08}, or Pach and Tardos \cite{PT09}. 

A breakthrough algorithm of Moser and Tardos was published in 2010 \cite{MoserTardos10}. It shows that after an initial sampling of all variables, a logarithmic number of resampling a maximal independent set of violated events is sufficient to compute a satisfying assignment.
Over the following decade, their framework was substantially developed in several directions: extending the analysis to stronger LLL criteria~\cite{KS11,Pegden14}, generalizing resampling to more complex probability spaces \cite{HV20}, studying the properties of the resulting distributions \cite{HSS11,HS17}, and developing deterministic and more efficient parallel algorithms \cite{HH17,CGH13,Harris17}.

\subsection{The Importance of LLL in Distributed Computing}\label{sec:lll-importance}

The LLL is inherently local in nature, replacing the global requirements of the union bound to ensure global correctness with conditions only involving local dependencies.  This makes it particularly well suited to the distributed LOCAL model, where algorithms must achieve global correctness through local computation and communication.

\paragraph{Completeness of LLL for sublogarithmic computation.}
In their seminal work \cite{CP19} Chang and Pettie formalized the importance of LLL for distributed algorithms. They showed that the LLL can be used to amplify the success probabilities of distributed algorithms. More precisely, they showed that any sublogarithmic-time distributed algorithm $\mathcal{A}$ with a suitable (relatively large) error probability at each node can be automatically turned into an algorithm with a \emph{high probability guarantee} by an efficient reduction to a single LLL instance. The objective of the LLL is to find random bits for $\mathcal{A}$ with which it does not err at any node. The resulting algorithm succeeds globally whenever the LLL succeeds, even if the original local error probabilities are too large for a union bound. See \Cref{thm:LLL-speedup} for more details and the guarantees obtained by using \Cref{thm:main} in this reduction. 

As an immediate corollary they obtain that the constructive LLL problem is complete for sublogarithmic computation for all so-called locally checkable labeling problems (LCLs). This yields a fundamental gap in the randomized complexity landscape of LCLs:  Any LCL problem either requires $\Omega(\log n)$ rounds or it can be solved in the same time as the constructive LLL with some polynomial criterion. 

\emph{LCLs} are central to the field, and example problems in the class are the maximal independent set problems, many vertex and edge coloring problems, but also many edge-orientation problems. Informally, an LCL problem is a problem on constant-degree graphs whose solution can be verified in constant time. The problem class was introduced in the 90's by Naor and Stockmeyer~\cite{NS95} and after the seminal work of \cite{CP19} countless publications have investigated the complexity landscape of these problems, e.g., \cite{CP19,CKP19,RG20,BBOS20,BBOS21,C20,GRB22}.

\paragraph{Distributed LLL algorithms.}
The Moser--Tardos framework yields an $O(\log^2 n)$-round distributed algorithm under criterion $(1+\eps)ep(d+1)<1$.
Chung, Pettie, and Su improved this to $O(\log n)$ when $(1+\eps)epd^2<1$ \cite{CPS17}.

Fischer-Ghaffari's algorithm, combined with Rozhoň–Ghaffari’s deterministic network decomposition \cite{RG20}, runs in $O(d^2) + \poly(\log\log n)$ rounds \cite{FG17}, for criteria $p(ed)^{32} < 1$.
Davies \cite{Davies_soda23} improved the first term to $d/\log d$.

Degree-independent $\poly(\log\log n)$-round algorithms were previously known under an exponential criterion \cite{GHK18,RG20}, extended to weaker criteria by Davies  \cite{Davies_soda23}, and such algorithms have also been developed for specific applications such as 
slack generation~\cite{HM24,HMP24}, degree splitting~\cite{HMN_disc22,Davies_soda23}, coloring triangle-free graphs~\cite{HMP24,DZ26}, and coloring with fewer than $\Delta+1$ colors~\cite{FHJM26}. 

\paragraph{Lower Bound and the Chang-Pettie Conjecture.}
In \cite{BFHKLRSU_stoc16}, the authors gave an $\Omega(\log_{\Delta}\log n)$ lower bound for LLL in randomized LOCAL, even when $p2^d \leq 1$.
Chang and Pettie conjectured $O(\log\log n)$ was the true complexity under polynomial criterion on bounded-degree graphs. This holds for LLLs with tree-structured dependency graphs even in the case of unbounded degrees \cite{CHLPU20}.

\paragraph{Other Models.}
LLLs have also been studied in other models: the bandwidth-constrained CONGEST model \cite{MU21,HMP24}, CONGESTED-CLIQUE and MPC \cite{Davies-Peck25}, and the Local Centralized Algorithm (LCA) model \cite{BGR21,achlioptas2020,Davies-Peck25}.

\subsection{Relation to previous techniques}
Our proof is based on exponent amplification for the LLL, which we achieve
by controlling conditional probabilities under retractions. Both ideas
are inspired by prior work, but we develop them quite differently here.
Next, we discuss the precedents for each idea and how our approach differs.

\paragraph{Exponent amplification.}
Fischer and Ghaffari~\cite{FG17} introduced exponent amplification as an
approach to building faster algorithms for the LLL. They first developed
a sublogarithmic LLL algorithm whose running time improves with the
criterion exponent (in an exponential way). They then viewed this algorithm's run
as a new LLL with a much larger criterion exponent, borrowing the idea from
Chang and Pettie's completeness result~\cite{CP19}, outlined in
Section~\ref{sec:lll-importance}. This enabled their base algorithm to
solve the new instance significantly faster.
Ghaffari, Harris, and Kuhn~\cite[Sections~5.3--5.4]{GHK18} formalized
this bootstrapping principle, extended the fast algorithms to a larger
range of degrees, and iterated the amplification to further improve the
dependence on $n$.
Our exponent-amplification lemma is inspired by this perspective.
Our lemma requires no upper bound on the dependency degree, and our amplification
transforms an LLL with $p=d^{-\lambda}$ into a new one with
$p'=(d')^{-\lambda'}$ in $O(1)$ rounds, with
$\lambda'=\Omega(\lambda^2)$, while passing only to $G'=G^{O(1)}$.
This quadratic exponent amplification is the key driver of our result.

\paragraph{Conditional probabilities under retractions and dangerous events.}
Our exponent amplification is built on the idea of controlling how partial retractions
change event probabilities, particularly through notions of dangerous
events whose marginals under retraction can become too high.
Similar ideas have been used before.
Fischer--Ghaffari~\cite{FG17} gradually set variables and freeze remaining variables of an event when its conditional probabilities become too
large. Chung--Pettie--Su~\cite{CPS17} and
Ghaffari--Harris--Kuhn~\cite{GHK18} directly consider partial retractions
and how they may increase the conditional probability of a bad event.
Their approaches must handle all $2^d$ possible retraction patterns around
an event, leading to speedups for LLL instances satisfying exponential criteria.
Ghaffari, Harris, and Kuhn~\cite{GHK18} and Davies~\cite{Davies_soda23}
further studied conditional event probabilities under certain structural
retractions through the notions of \emph{fragility} and \emph{resilience},
respectively; these are particularly useful for events involving sums of
independent random variables.

Our construction restricts the number of relevant local retraction patterns
with a local independent-packing constraint.
Specifically, our definition of dangerousness considers only retractions generated by at most $k$ nearby centers,
giving $d^{O(k)}$ possible retractions.
The LLL $\Pi'$ forbids independent packing of $k+1$ dangerous events
in a constant-radius neighborhood (Condition~\eqref{eq:target}).
Avoiding the bad events of $\Pi'$ ensures that every retraction in the repair step is covered  by our definition of dangerousness and that the resulting
repair LLL has dependency degree $O(k)$.
Crucially, a violation of $k$ or more independent nearby constraints requires $k$ independent
dangerous events, which pushes this $k$ into the exponent and yields a
bad-event probability of at most $d^{-\Omega(\lambda k)}$ in $\Pi'$
(see \cref{eq:amplified-probability}).
Choosing $k=\Theta(\lambda)$ therefore yields quadratic exponent amplification
and a repair instance satisfying an exponential criterion.
We note that we cannot take $k$ much larger: the union bound over
$d^{O(k)}$ retraction patterns limits us to $k=O(\lambda)$.
Moreover, $k$ bounds the dependency degree of the repair LLL, which must
satisfy the exponential criterion.

\subsection{Further Corollaries}
Many algorithms adapt combinatorial results to the LOCAL model. In the absence of a general sub-logarithmic LLL algorithm, the technique was to carefully analyze each LLL, prove that it shatters and that post-shattering instances can be resolved efficiently. \cref{thm:main} supplies a general solver when the application satisfies the criteria: $(\Delta-k_\Delta+1)$-coloring \cite{FHJM26}, triangle-free graphs \cite{PS15,HMP24,DZ26}, degree splitting \cite{HMN_disc22}, frugal coloring \cite{CPS17}, and coloring locally sparse graphs \cite{AKS_jct99,Vu_jgt99}.

The first direct corollary, by using the speed-up technique from \cite{CP19}, is a technique for attenuating the error probability of LOCAL algorithms. Prior to this work, the higher cost of LLL introduced additional dependency on $\Delta$.
\begin{theorem}[{\cite[Theorem 4.1]{CP19} with \cref{thm:main}}]
    \label{thm:LLL-speedup}
    Let $\Delta \geq 3$ and $\mathcal{A}$ be a $T(n,\Delta)$-round randomized LOCAL algorithm such that each vertex can check in $R(n,\Delta)$ rounds if its output is correct and every node is incorrect with probability at most $p$. If $p\Delta^{2000(T(n,\Delta)+R(n,\Delta))} < 1$, there is a $(T(n,\Delta) + R(n, \Delta)) \cdot \poly(\log \log n)$ round algorithm that produces an output of $\mathcal{A}$ such that all nodes are correct with high probability.
\end{theorem}
\begin{proof}
    All ingredients of the proof are from \cite[proof of Theorem 4]{CP19}. We omit the last steps of their argument, which relies on the condition that $T(n,\Delta) \leq f(\Delta) + \eps \log_\Delta n$ for a sufficiently small constant $\eps$, which we do not assume.
    Let failures of nodes in $\mathcal{A}$ be the bad events. Failure of a node is a function of outputs within distance $R(n,\Delta)$, and the output of each node is a function of its distance-$T(n,\Delta)$ neighborhood. As a result a bad event only shares variables with at most $\Delta^{2(T(n,\Delta)+R(n,\Delta))}$ other bad events.
    If $p\Delta^{2000(T(n,\Delta)+R(n,\Delta))} < 1$, we have an LLL s.t.\ $p < d^{-1000}$, an LLL solvable using \cref{thm:main} in $\poly(\log \log n)$ rounds with high probability. The $(T(n,\Delta) + R(n, \Delta))$ multiplicative factor is the overhead of simulating communication over the LLL graph using communication over the original graph.
\end{proof}

For graphs of constant degree and LCLs (for which $R(n,\Delta) = R(\Delta)$) \cite{CP19} uses an additional argument to replace the multiplicative dependency in $T(n,\Delta)$ by a term that is a pure function of $\Delta$ when $T(n,\Delta) \in o(\log_\Delta n)$.
Combined with earlier algorithms for LLL~\cite{FG17,RG20,Davies_soda23}, this already implied a gap between $\poly(\log\log n)$ and $O(\log n)$ in the complexity hierarchy of LCLs. Our result merely reduces the dependencies in $\Delta$ hidden in the $\poly(\log\log n)$.

\Cref{thm:LLL-speedup} only restates a theorem by Chang and Pettie when plugging in the complexity of our algorithm for LLL. As the complexity of earlier algorithms for LLL with a polynomial criterion grows significantly with the dependency degree $d$ of the LLL, we obtain a much better complexity when $d$ is large for the LLL that appears in the proof. This is true for many regimes of $T$, $R$, and $\Delta$.

\begin{theorem}
    \label{thm:error-attenuation}
    Suppose there exists a randomized LOCAL algorithm to solve a problem with error probability $n^{-1000}$ in $T(n)$ rounds and whose output can be checked in $T(n)$ rounds. Then, for all $\eps(n) \in (0,1)$, there exists an algorithm for solving the same problem in  $T(n) \poly(\log \log (n/\eps))$ rounds and with error probability at most $\eps(n)$.
\end{theorem}
For example, this allows us to solve problems with error probability $\eps = 2^{-\log^{\gamma} n}$
for any constant $\gamma$ of one's choice in $T(n)\poly(\log \log n)$ rounds.

\begin{proof}
    When considering an LLL where each bad event stems from a node, the dependency degree is trivially bounded by $n-1$. With an error probability $\leq n^{-1000}$, \cref{thm:main} applies. By running the algorithm of \cref{thm:main} using $N=\lceil n/\eps(n) \rceil$ as upper bound on the size of the graph, one obtains error probability at most $1/N \leq \eps(n)$.
\end{proof}

\paragraph{Edge Coloring.}
\cite{CHLPU20} showed how to compute edge colorings with $\Delta+\tilde{O}(\sqrt{\Delta})$ colors by sequentially invoking $O(\log \Delta)$ LLLs with a polynomial criterion; these can be solved in $\poly(\log\log n)$ rounds via \Cref{thm:main}.
\begin{theorem}[{\cite[Theorem 4]{CHLPU20} with \cref{thm:main}}]
There is a $\poly(\log\log n)$-round randomized LOCAL algorithm that w.h.p.\ computes a $\Delta+O(\sqrt{\Delta}\log^3 \Delta)$-edge coloring of any graph with maximum degree $\Delta\leq \log^{O(1)}n$. 
\end{theorem}

\paragraph{Vertex Degree Splitting.}

A $k$-\emph{vertex splitting} of discrepancy $\delta$ of a graph $G=(V,E)$ is a partition of its vertices $V$ into $k$ sets $V_1,\dots,V_k$ s.t.\ for each node $v \in V$ and $i \in [k]$, $\abs{\card{N(v) \cap V_i} - \deg(v)/k} \leq \delta$. That is, the neighbors of each node are split roughly equally between the $k$ parts.
\begin{theorem}
\label{thm:splitting}
For any $k\in[2,\Delta/(6400\ln \Delta)]$, there is a $\poly(\log \log n)$-round randomized LOCAL algorithm that w.h.p.\ computes a $k$-vertex splitting with discrepancy $80 \sqrt{(\Delta/k)\ln \Delta}$.
\end{theorem}

In particular, if splitting into two colors ($k=2$), \cref{thm:splitting} results in the nodes having 
two-way discrepancy $O(\sqrt{\Delta\log\Delta})$
in $\poly(\log \log n)$ time.
The same fine 2-split runs in time $O(\Delta/\log \Delta) + \poly(\log \log n)$ in prior work~\cite[Thm.~19]{Davies_soda23}.

\begin{proof}
Let each node join a bucket picked uniformly at random. In expectation, a node $v$ has $\deg(v)/k$ neighbors joining each bucket. For any $\eps \in (0,1]$, a Chernoff bound bounds the probability of getting a deviation greater than $\pm \eps\Delta/k$ by $\leq 2\exp(-\eps^2 \Delta/(3k))$.
We set $\eps$ of order $\sqrt{(k/\Delta)\log \Delta}$.

The bad event at each node is to have a large discrepancy for one of the buckets. This LLL has a dependency degree of $d\leq \Delta + \Delta(\Delta-1) = \Delta^2$.
The maximum probability $p$ of a bad event is bounded by $2k\exp(-\eps^2 \Delta/(3k))$, and we need $p < d^{-1000}$ to apply the algorithm for LLL.
By setting $\eps = 80\sqrt{(k/\Delta) \ln \Delta}$, which is $\leq 1$ per the maximum value of $k$, one gets $p \leq 2 \Delta \exp(-(80^2 /3) \ln \Delta) < \Delta^{-2000} \leq d^{-1000}$. The discrepancy is limited to $\eps \Delta/k = 80\sqrt{(\Delta/k) \ln \Delta}$.
\end{proof}

\paragraph{Domatic partition.}
A $k$-\emph{domatic partition} of a graph $G=(V,E)$ is a partition of $V$ into $k$ disjoint sets $V_1,\dots,V_k$ s.t.\ each $V_i$ is a dominating set of $G$. That is, for every node $v$ and $i \in [k]$, either $v\in V_i$ or $N(v) \cap V_i \neq \emptyset$.
The \emph{domatic number} $\dom(G)$ of a graph is the largest $k$ s.t.\ $G$ admits a $k$-domatic partition. With $\delta$ and $\Delta$ respectively the minimum and maximum degree in $G$, trivially $\dom(G) \leq \delta+1$, and it is known that $\dom(G) \geq (1-o(1)) \delta / \ln \Delta$~\cite{FHKS_sicomp02_domatic}.
\begin{theorem}
\label{thm:domatic}
Let $G$ be of minimum degree $\delta$, maximum degree $\Delta$, and let $k\leq \delta/(2001\ln\Delta)$. There exists a $\poly(\log \log n)$ randomized LOCAL algorithm for computing a $k$-domatic partition of $G$ with high probability.
\end{theorem}
\begin{proof}
    If each node picks a color at random, the probability that no neighbor of a node $v$ picks a particular color is at most $(1-1/k)^{\deg(v)} \leq e^{-\deg(v)/k} \leq e^{-\delta/k}$.
    The probability that there exists a color absent from $v$'s neighborhood is at most $ke^{-\delta/k}$. For each node $v$, let this be the bad event at $v$, defining an LLL of probability $p \leq ke^{-\delta/k}$ and dependency degree $d \leq \Delta^2$.

    With $k \leq \delta/(2001 \ln \Delta)$, one gets $p \leq \Delta^{-2000} \leq d^{-1000}$. This makes it possible to use the $\poly(\log \log n)$-round LLL algorithm.
\end{proof}
Note that the number of parts of \Cref{thm:domatic} is within a multiplicative constant of what is existentially always possible.

\section*{Acknowledgements}

This work was supported in part by the Research Council of Finland, Grant 363558.
This research was funded in whole or in part by the Austrian Science Fund (FWF) \url{https://doi.org/10.55776/P36280}, \url{https://doi.org/10.55776/I6915}. For open access purposes, the author has applied a CC BY public copyright license to any author-accepted manuscript version arising from this submission.

MMH was supported by the Icelandic Research Fund, Grant 2511609.

\bibliographystyle{alphaurl}
\bibliography{ref-preamble,references}

\appendix
\section{Proofs of the Preliminary Lemmas}\label[appendix]{sec:prelim-proofs}

\subsection{Ruling sets under bounded local independence}\label[appendix]{sec:proof-ruling}

\begingroup\restatinglemmatrue
\renewcommand{\thetheorem}{\ref*{lem:ruling}}
\LemmaBoundedIndependenceRulingSets*
\endgroup
\begin{proof}[Proof of \Cref{lem:ruling}]
We initialize $S_0$ to be an independent $\beta$-ruling set for
$\beta=O(\log\log n)$. This can be computed using
Gfeller--Vicari~\cite{GV_RulingSetBoundedGrowth_podc07} and
Schneider--Elkin--Wattenhofer~\cite{SEW13} in $O(\log\log n)$ rounds
with high probability.
For $i = 0, 1, 2, \ldots, \beta-1$, we do the following.
Each $s\in S_i$ locally chooses a maximal independent set $U_s$ in
$H[\B_H(s,3)\setminus\B_H(S_i,1)]$, using information from $\B_H(s,4)$.
Put $Q=\bigcup_s U_s$, compute a maximal independent set $I$ of $H[Q]$, and let
$S_{i+1} = S_i\cup I$.

Each iteration preserves independence, since $Q\cap\B_H(S_i,1)=\varnothing$. After the $i$-th iteration, every vertex has $\dist_H(v, S_i) \leq \max\set{2, \beta - i}$. Let $v$ be a vertex with $\dist_H(v,S_i) \geq 3$ and $u$ the vertex at distance exactly $3$ from $S_i$ on the shortest path from $v$ to $S_i$. At the start of the $(i+1)$-th iteration, $u$ is within distance $1$ of $U_s$ for some $s\in S_i$; hence, within distance $2$ of $I$. Then $\dist_H(v, S_{i+1}) \leq \dist_H(v,u) + 2 = \dist_H(v, S_i) - 1 \leq \beta - (i+1)$. In particular, every vertex is at distance at most 2 from $S_\beta$.

To bound the round complexity, note that each $U_s$ has at most $k$ members.
For $v\in Q$, a member of $\B_H(v,1)\cap Q$ can be proposed only by
a center in $\B_H(v,4)$. There are at most $k$ such centers, since
$S_i$ is independent, so
$H[Q]$ has degree at most $k^2-1$. Therefore, the MIS algorithm of
\cite[Theorem 1.2]{Ghaffari_soda16}, with the finishing step of
\cite[Corollary 3.2]{RG20}, computes $I$ in
$O(\log(k+2))+\poly(\log\log n)=\poly(\log k, \log\log n)$ rounds.
A flaw in the shattering analysis of \cite{Ghaffari_soda16} was rectified in
\cite[Appendix B.2]{GHMN26}.
Other work takes $O(1)$ rounds per iteration; the total over
$\beta = O(\log\log n)$ iterations is $\poly(\log k, \log\log n)$.
\end{proof}

\subsection{Exponential-criterion LLL}\label[appendix]{sec:proof-exponential}
\begingroup\restatinglemmatrue
\renewcommand{\thetheorem}{\ref*{lem:exponential}}
\LemmaExponentialLLL*
\endgroup

\begin{proof}[Proof of \Cref{lem:exponential}]
Fix a desired failure exponent $a>0$. If $np\le n^{-a}$, one sample suffices by a union bound. Otherwise $d<(a+1)\log_2 n/C=O(\log n)$.

We use the terminology and procedure of
\cite[Definition 6.1 and Algorithm 2]{GHK18}.
For a partial assignment $y$, the \emph{marginal probability}
$\Prb_y[A_v]$ is the probability of $A_v$ when assigned variables are
fixed as in $y$ and unassigned variables are sampled independently
from their original distributions.
The \emph{retraction} $\ret(x,U)$ of a full assignment $x$ with respect to
an event set $U$ erases $\bigcup_{u\in U}V_u$ and keeps $x$ elsewhere.
Put $q=(ed)^{-3}$.
In this proof, $v$ is \emph{$q$-dangerous}\footnote{This definition of dangerous is different from the one in \Cref{sec:proof}.} under $x$ if its retracted-marginal
$\Prb_{\ret(x,U)}[A_v]\ge q$ for some $U\subseteq\B_G(v,1)$.
This includes $U=\varnothing$, so nodes of violated events are dangerous. 

Assign each variable a fixed incident vertex as owner and draw a random
assignment $X$ using independent randomness at the owners.
Let $M$ be the set of $q$-dangerous vertices under $X$ and set $Y=\ret(X,M)$.
The \emph{residual instance} consists of the events indexed by
$\mathcal A=\B_G(M,1)$, with the assigned entries of $Y$ fixed.
In the residual instance, erased variables have their original independent
distributions, while unerased variables retain their values in $Y$.
For every $v$, retraction by $M$ affects $A_v$ exactly as retraction by
$M\cap\B_G(v,1)$. Thus, $\Prb_Y[A_v]\le q$ for every $v$: nondangerous events
satisfy this by definition, while dangerous events have all their
variables erased. Hence, dangerous events have probability at most $p\le q$.
Every event outside $\mathcal A$ is permanently unviolated: none of its variables
is erased, and $\mathcal A$ contains all violated events.

For fixed $U\subseteq\B_G(v,1)$,
$\E[\Prb_{\ret(X,U)}[A_v]]=\Prb[A_v]\le p$.
Markov's inequality and a union bound over at most $2^{d+1}$ choices
give $\Prb[v\in M]\le2^{d+1}p/q$.
Consequently, $\Prb[v\in\mathcal A]\le(d+1)2^{d+1}p/q\le(ed)^{-8}$
for sufficiently large absolute $C$.
The event that $v\in M$ is determined by $X_{V_v}$, so membership of $v$ in
$\mathcal A$ depends only on owners' randomness within $\B_G(v,2)$.
The membership events $v\in\mathcal A$ for vertices at pairwise distance
at least $4$ are mutually independent, since their variable sets
$\bigcup_{u\in\B_G(v,1)}V_u$ are pairwise disjoint.
By the standard shattering bound with radius parameter $2$,\footnote{We use the corrected argument in \cite[Corollary 1]{GHMN26}, replacing \cite[Theorem 5.1]{GHK18}. Here genuine two-hop randomness dependence gives the required joint-probability bound for $4$-independent sets.}
every component of $G[\mathcal A]$ has at most
$N=O(d^4\log n)=\poly(\log n)$ vertices with probability at least
$1-n^{-a}$.

Each residual LLL has degree at most $d$ and marginal probabilities
at most $q$, with $e(d+1)q<1/2$.
Any erased variable shared by two affected events joins them in
$G[\mathcal A]$, so different components share no erased variables.
Thus, \cite[Corollary 3.8]{RG20}, using size bound $N$, solves all
components in parallel in $\poly(\log\log n)$ rounds, allowing original IDs by
\cite[Remark 2.10]{RG20}. All preceding steps take $O(1)$ rounds.
\end{proof}

\end{document}